\documentclass{llncs}
\usepackage{proof}
\usepackage{amsmath}
\usepackage{amssymb}
\usepackage{graphicx}

\title{Constructing Independently Verifiable Privacy-Compliant Type Systems for Message Passing between Black-Box Components}
\titlerunning{Independently Verifiable Privacy-Compliant Type Systems}
\author{Robin Adams\inst{1} \and Sibylle Schupp\inst{2}}
\institute{Chalmers University of Technology \\ \email{robinad@chalmers.se} \and Technische Universit\"{a}t Hamburg-Harburg \\
\email{sibylle.schupp@tuhh.de}}

\newtheorem{algorithm}{Algorithm}

\newcommand{\Ag}{\ensuremath{\mathtt{Ag}}}
\newcommand{\Safe}[2]{\ensuremath{\mathtt{Safe} \left( {#1} , {#2} \right)}}
\newcommand{\Pos}[2]{\ensuremath{\mathtt{Pos} \left( {#1}, {#2} \right)}}

\begin{document}
\maketitle

\begin{abstract}
 \emph{Privacy by design} (PbD) is the principle that privacy should be considered at every stage of the software engineering process.  It is increasingly both viewed as best practice
 and required by law.  It is therefore desirable to have formal methods that provide guarantees that certain privacy-relevant properties hold.  We propose an approach that
 can be used to design a privacy-compliant architecture without needing to know the source code or internal structure of any individual component.

We model an
architecture as a set of \emph{agents} or \emph{components} that pass
\emph{messages} to each other.  We present in this paper algorithms that take as input an architecture and a set of privacy constraints, and output an extension of the original architecture
that satisfies the privacy constraints.
\end{abstract}

\section{Introduction}

\emph{Privacy by Design} is the principle that privacy should be a consideration at every stage of the software design process\cite{cavoukian2012privacy}.
It is increasingly seen as best practice for privacy protection, including by the International Conference of Data Protection and Privacy Commissioners\cite{cavoukian2010resolution}
and the US Federal Trade Commission\cite{federal2012protecting}, and is a legal requirement in the EU since the General Data Protection Regulation (GDPR) came into force
on 25 May 2018\cite{eu:gdpr}.

It is therefore desirable to create methods that will provide a guarantee that software satisfies certain privacy-relevant properties.  To this end, a substantial amount
of research (both formal methods and other approaches)
has been devoted to this problem, including static analysis of source code (e.g.~\cite{ferrara2015morphdroid},\cite{cortesi2015datacentric}); real-time ``taint tracking'' of the data released by apps on
a mobile device (e.g.~\cite{schreckling2013kynoid},\cite{enck2014taintdroid}); refinement techniques that preserve privacy properties as we refine in stages from a high-level
design to code (e.g.~\cite{alur2006preserving},\cite{clarkson2010hyperproperties}); or the creation of new programming languages which include representations of privacy-relevant properties
in types or annotations (e.g.~\cite{pottier2003information},\cite{myers2001jif}).

We can thus design a privacy-safe application, or verify that a given application is privacy-safe, provided we can access and/or change its source code.
However, in practice, many systems involve the interaction of different components, each controlled by a different person or organisation.  The source code might not be available, or
it might not be possible for us to change it.  New versions of each component may come out regularly, so that a privacy analysis we did using an old component quickly becomes obsolete.


In this paper, we will show how we can design
a type system for the \emph{messages} that the components pass to each other, in such a way that we can formally prove that, if every message passed is typable under this typing system,
then the privacy property must hold.  We indicate how an existing unsafe component can be adapted into a component that uses this typing system by providing each component with an \emph{interface}
through which all messages must pass, without needing to read or modify the component's source code.

The structure of the paper is as follows.  In Section \ref{section:example}, we give a relatively simple but realistic example of privacy constraints
that we may wish to hold, and show the architecture that our algorithms generate.  In Section \ref{section:arch}, we provide the formal definition of architecture that we use.
In Section \ref{section:firstalgo}, we define the algorithm for a simple constraint language and prove it correct.  In Section \ref{section:secondalgo}, we do the same development
again for a stronger language of constraints, of the form $\alpha \ni A \Rightarrow \beta \ni B$ (`if $\alpha$ possesses a term of type $A$
then $\beta$ must previously have possessed a term of type $B$').  Finally we survey some related work in Section \ref{section:related}, and conclude in Section \ref{section:conclusion}.

\section{Motivating Example}
\label{section:example}

We now give an example of realistic privacy constraints that we might wish to introduce, and the architectures that are produced by our algorithms.
The example is similar to an example considered by Barth et al. \cite{contextualintegrity}.

The US Children's Online Privacy Protection Act (COPPA) includes the clause:

\begin{quote}
  When a child sends protected information to the website, a \emph{parent} must have previously received a privacy notice from the web site operator, [and] granted consent to the web site operator.
\end{quote}

We propose to model a system as being composed of \emph{agents} or \emph{components} who pass \emph{messages} to each other.
The possible messages are provided by a \emph{type system}, which consists of a set of \emph{types} and a set of \emph{constructors}.
These two sets determine the set of \emph{terms}, each of which has a type. We write $t : A$ to denote that the term $t$ has type $A$.

A \emph{message} is a triple ($\alpha$, $t$, $\beta$), where
$\alpha$ and $\beta$ are agents and $t$ is a term; this represents the agent $\alpha$ sending the piece of data $t$ to $\beta$.  If $t : A$,
then we write this message as $\alpha \stackrel{t}{\rightarrow} \beta$ or $\alpha \stackrel{t : A}{\rightarrow} \beta$.

For the COPPA example, Figure \ref{fig:coppa1} suggests an architecture with three agents, Child, Website, and Parent.
In the initial state, Child possesses a term info : INFO, Website possesses policy : POLICY, and Child may send messages of type INFO to Website, etc.
This represents a website which can send its privacy policy to the parent; the parent may send consent for the website to collect the child's protected info;
and the child may send their protected info to the website.
However, at the moment, there is nothing to prevent the protected info being sent to the website without either policy or consent having been sent.

Formally, an \emph{architecture} is described by specifying the following (see Definition \ref{df:arch}):
\begin{itemize}
\item
for any agent $\alpha$, which constructors an agent possesses in the initial state;
\item
for any two agents $\alpha$, $\beta$, the set of types $A$ such that $\alpha$ may pass a message of type $A$ to $\beta$.
\end{itemize}
If $A$ is a type, we shall sometimes say '$\alpha$ can send $A$ to $\beta$' to mean '$\alpha$ may send messages of type $A$ to $\beta$'.

We envision the designer beginning with a set $\Ag$ of agents and a type system $\mathcal{T}$ which describes the pieces of data they are interested in.
They write down the set $\mathcal{C}$ of privacy constraints that they wish the finished system to have.  For now, we consider constraints of these two forms
(see Definitions \ref{df:constraint} and \ref{df:constraint2}):
\begin{itemize}
 \item $\alpha \ni A \Rightarrow B$: If agent $\alpha$ has a piece of data of type $A$, then a piece of data of type $B$ must have previously been created.
 \item $\alpha \ni A \Rightarrow \beta \ni B$: If agent $\alpha$ has a piece of data of type $A$, then agent $\beta$ must previously have had a piece of data of type $B$.
\end{itemize}

The privacy constraints that we require for the architecture in Figure \ref{fig:coppa1} include

$$Website \ni INFO \Rightarrow Website \ni CONSENT$$
$$Website \ni CONSENT \Rightarrow Parent \ni POLICY$$

The first constraint specifies that agent Website possesses INFO only if it previously has received data of type CONSENT.
The second constraint specifies that agent Website possesses CONSENT only if the Parent agent has received the POLICY before.
(We will add a third constraint later, in Section \ref{section:revisit}.)

\begin{figure}
 \includegraphics[width=0.95\textwidth]{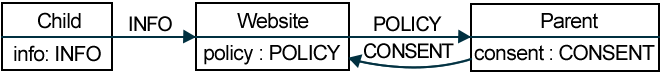}
 \caption{An Architecture That Allows Privacy Breach}
 \label{fig:coppa1}
\end{figure}

Given privacy constraints, we show how to extend $\mathcal{T}$ to a type system $\mathcal{T}_\mathbb{C}$.
The type system $\mathcal{T}_\mathbb{C}$ includes a set of new types $C_\alpha(A)$.  A term of type $C_\alpha(A)$ is called
a \emph{certified} term.  As well as the plain INFO type,
for example, the safe architecture contains the type $C_{Website}(INFO)$.  A term of this type represents
a piece of data from which $Website$ can extract a term of type $INFO$, but no other agent can.\footnote{In practice, this would presumably be achieved
by encryption, but we abstract from these implementation details here.  See Section \ref{section:implementation} for more discussion.}
There are no restrictions on which agents may receive them or send certified terms.

The type system $\mathcal{T}_\mathbb{C}$ also has types $P_\alpha(A)$, and constructors $p_\alpha$ that construct terms of type $P_\alpha(A)$.
We may think of a term of type $P_\alpha(A)$ as a \emph{proof} that $\alpha$ possesses a term of type $A$.

The architecture created by our Algorithm 2 is shown in Fig. \ref{fig:coppa2}.  (For space reasons, we have listed only some of the constructors and messages, and
omitted the subscripts on the types $C_\alpha(A)$ and $P_\alpha(A)$.)  The algorithm creates new components $\mathit{IWebsite}$, the \emph{input interface} to $\mathit{Website}$, and
$\mathit{OWebsite}$, the \emph{output interface} for $\mathit{Website}$; and similarly input and output interfaces for $\mathit{Parent}$ and $\mathit{Child}$.

The constructor $p_{POLICY}$ takes a term of type $POLICY$ and constructs a term of type $P_{Parent}(POLICY)$ --- a proof that $Parent$ has received a term of type $POLICY$.
The constructor $m_{CONSENT}$ constructs a certified term of type $C_{Website}(CONSENT)$ out of a term of type $CONSENT$, plus the proof that the preconditions for $Website$ to be
allowed to read a term of type $CONSENT$, namely a term of type $P_{Parent}(POLICY)$.  The constructor $\pi_{CONSENT}$ then extracts the term of type $CONSENT$ from the certified term.
Similar comments hold for $m_{POLICY}$ and $\pi_{POLICY}$, and the other new constructors in Fig. \ref{fig:coppa2}.

It can be seen that, while $\mathit{Child}$ may send $\mathit{INFO}$ to $O\mathit{Child}$ at any time, the only way for
the data to travel any further is for a term of type $C_{\mathit{Website}}(\mathit{INFO})$ to be created; this can only happen if a term of type $P_{\mathit{Website}}(\mathit{CONSENT})$ has been created;
this can only happen if a term of type $\mathit{CONSENT}$ reaches $O\mathit{Website}$; and this can only happen if $\mathit{Website}$ has a term of type $\mathit{CONSENT}$.  Similar considerations hold
for our other negative constraint.

\begin{figure}
 \includegraphics[width=0.95\textwidth]{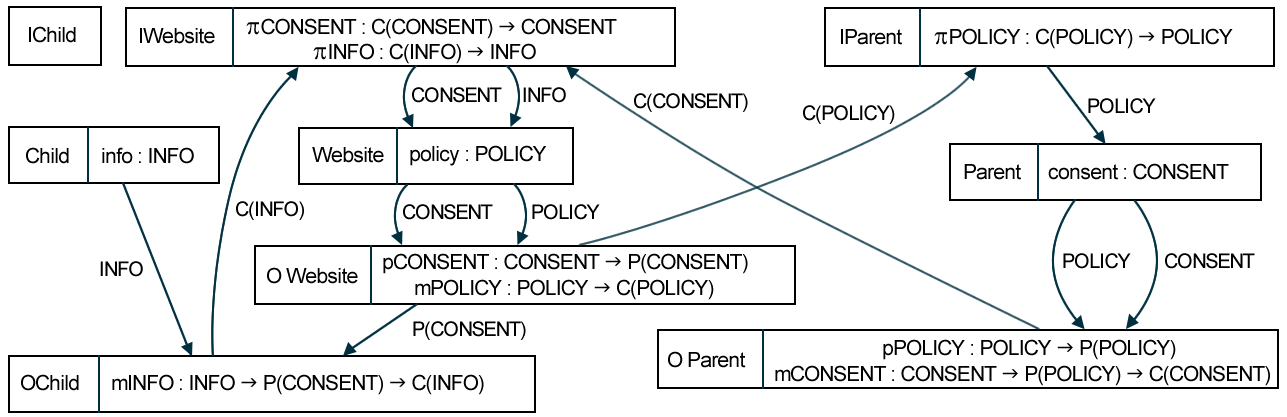}
 \caption{A Privacy-Safe Architecture}
 \label{fig:coppa2}
\end{figure}


We can partition the agents in Fig. \ref{fig:coppa2} into three sets:
$\{ Child, IChild, OChild \}$, $\{ Website, IWebsite, OWebsite \}$, $\{ Parent, IParent, OParent \}$.
Each set thus consists of one of the agents from Fig. \ref{fig:coppa1}, plus its two new interfaces.  Note that, if an agent from one set passes a message to an agent in another
set, then that message has type $C_\alpha(A)$ or $P_\alpha(A)$ for some $\alpha$, $A$.  In the rest of this paper, we will prove two results (Theorems \ref{theorem:safeone} and \ref{theorem:safetwo})
that give general conditions such that, if an architecture can be partitioned in a way that satisfies these conditions,
then a given set of privacy constraints are satisfied.

\subsection{Note on Implementation}
\label{section:implementation}

In practice, certification on the one hand, access on the other hand, could be implemented through encryption and decryption.
But, other mechanisms possibly exist as well.  The type systems we present in this paper abstract from these details.  They specify which agents
may and may not access which data, without specifying how this is to be done.

The terms of type $P_\alpha(A)$ should, in practice, ideally be an appropriate zero-knowledge proof which guarantees that $\alpha$ possesses a term of type $A$,
without revealing the value of the term of type $A$.  Again, in this paper we abstract from the details of how this would be implemented.

However, we expect it to be possible to implement these types in such a way that the designer could publish both the set of constraints $\mathbb{C}$ and the
type system $\mathcal{T}_\mathbb{C}$, and an independent third party (the user, a regulatory authority, or anyone else) to verify both that our algorithm maps
$\mathbb{C}$ to $\mathcal{T}_\mathbb{C}$, and that any given message is typable under $\mathcal{T}_\mathbb{C}$.  This would greatly increase the trust that all parties
can have that the global privacy policies $\mathbb{C}$ hold true.

We also note that, if there are large numbers of agents in our system, we will need a large number of types.  In our motivating example, if we have many children and many parents, then
we will need types $C_{CHILD_1}(CONSENT)$, $C_{CHILD_2}(CONSENT)$, etc. and a way to ensure that $\pi_{\alpha A}$ accepts terms of type $C_{\alpha'}(A)$ only if $\alpha = \alpha'$,
requiring the use of dependent types.  For now, this is left as work for the future.

\section{Architectures}
\label{section:arch}

We now describe the language we use for specifying architectures.  This system was inspired by work by le M\'{e}tayer et al \cite{pbd} and Barth et al \cite{contextualintegrity}.

An \emph{architecture} consists of \emph{agents} who pass \emph{messages} to each other.  Each message is a term that can be typed in a \emph{type system}.

\begin{definition}[Type System]
 A \emph{type system} is given by the following:
 \begin{itemize}
  \item A set of \emph{atomic types}.  The set of \emph{types} is the defined inductively by:
  \begin{itemize}
  \item Every atomic type is a type.
  \item If $A$ and $B$ are types, then $A \rightarrow B$ is a type.
 \end{itemize}
 \item
 A set of \emph{constructors}, each with an associated type.
 \end{itemize}
The set of \emph{terms} of each type is then defined inductively by:
\begin{itemize}
 \item Every constructor of type $A$ is a term of type $A$.
 \item If $s$ is a term of type $A \rightarrow B$ and $t$ is a term of type $A$, then $st$ is a term of type $B$.
\end{itemize}
We write $t : A$ to denote that $t$ is a term of type $A$.
\end{definition}

In the example in Figure \ref{fig:coppa1}, the atomic types are $INFO$, $CONSENT$ and $POLICY$.  The constructors are
$info$ which has type $INFO$, $policy$ which has type $POLICY$, and $consent$ which has type $CONSENT$.
In the example in Figure \ref{fig:coppa2}, the architecture has been extended with new atomic types such as $C_{Website}(INFO)$, and new constructors
such as $p_{POLICY}$, which has type $POLICY \rightarrow P_{Parent}(POLICY)$.

\begin{definition}[Architecture]
  \label{df:arch}
 Given a type system $T$, an \emph{architecture} $\mathcal{A}$ over $T$ consists of:
 \begin{itemize}
  \item a set $\Ag$ of \emph{agents} or \emph{components};
  \item for every agent $\alpha$, a set $H_\alpha$ of constructors that $\alpha$ \emph{initially possesses} or \emph{initially has};
  \item for every ordered pair of distinct agents $(\alpha, \beta)$, a set $M_{\alpha \beta}$ of \emph{atomic} types that $\alpha$ may send in a \emph{message} to $\beta$.
 \end{itemize}
 We shall write $\alpha \stackrel{A}{\rightarrow} \beta$ to denote that $A \in M_{\alpha \beta}$.
\end{definition}

(Note that only terms of atomic type can be passed between agents.)

In the example in section 2, we have $\Ag = \{ Child, Website, Parent \}$.  The agent $Child$ initially possesses the constructor $info$, and $Website$ initially possesses $policy$,
and $Parent$ initially possesses $consent$.  We have $M_{Child, Website} = \{ INFO \}$; thus, $Child$ may send messages of type $INFO$ to $Website$.  We also have
$M_{Website, Parent} = \{ POLICY \}$ and $M_{Parent, Website} = \{ CONSENT \}$.

We will use lower-case Greek letters $\alpha$, $\beta$, \ldots for agents, lower-case Roman letters $s$, $t$, \ldots for terms, and capital Roman letters $A$, $B$, \ldots for types.
The letter $c$ is reserved for constructors.

Let us say that an agent $\alpha$ can \emph{compute} terms of type $A$ iff it possesses a constructor of type $B_1 \rightarrow \cdots \rightarrow B_n \rightarrow A$ for some $B_1$, \ldots, $B_n$.

\begin{definition}
 Let $\mathcal{A}$ be an architecture.
 \begin{enumerate}
  \item An \emph{event} or \emph{message} is an expression of the form $\alpha \stackrel{t : A}{\rightarrow} \beta$, to be read as `$\alpha$ passes the term $t$ of type $A$ to $\beta$.'
  \item A \emph{trace} $\tau$ is a finite sequence of events.
  \item A \emph{judgement} is an expression of the form $\tau \vdash \alpha \ni t : A$, which we read as ``After the trace $\tau$, $\alpha$ has the term $t$ of type $A$.''
 \end{enumerate}
 We write $\tau_1, \tau_2$ for the concatenation of traces $\tau_1$ and $\tau_2$.
 We write $\tau_1 \sqsubseteq \tau'$ iff $\tau_1$ is a prefix of $\tau'$, i.e. there exists $\tau_2$ such that $\tau' = \tau_1, \tau_2$.
\end{definition}

The \emph{derivable} judgements are given by the rules of deduction in Figure \ref{fig:rules}.
We say that $\tau$ is a \emph{valid} trace through $\mathcal{A}$ iff $\tau \vdash \alpha \ni t : A$ is derivable for some $\alpha$, $t$, $A$.
We say that an agent $\alpha$ \emph{possesses} a term of type $A$ after $\tau$, and write $\tau \vdash \alpha \ni A$, iff there exists a term $t$ such that $\tau \vdash \alpha \ni
t : A$.
We say that there \emph{exists} a term of type $A$ after $\tau$, and write $\tau \vdash A$, iff $\tau \vdash \alpha \ni t : A$ for some $\alpha$, $t$.

\begin{figure}
\[ (init) \; \vcenter{\infer[(c : A \in H_\alpha)]{\vdash \alpha \ni c : A}{}} \quad
(message_1) \; \vcenter{\infer[(A \in M_{\alpha \beta})]{\tau, \alpha \stackrel{t : A}{\rightarrow} \beta \vdash \beta \ni t : A}{\tau \vdash \alpha \ni t : A}} \]
\[ (message_2) \; \vcenter{\infer[(A \in M_{\alpha \beta})]{\tau, \alpha \stackrel{t : A}{\rightarrow} \beta \vdash \gamma \ni s : C}{\tau \vdash \alpha \ni t : A \qquad \tau \vdash \gamma \ni s : C}} \]
\[ (func) \; \vcenter{\infer{\tau \vdash \alpha \ni ft : B}{\tau \vdash \alpha \ni f : A \rightarrow B \qquad \tau \vdash \alpha \ni t : A}} \]
\caption{Rules of Deduction}
\label{fig:rules}
\end{figure}

The rule (init) states that, if $\alpha$ initially possesses $c$, then $\alpha$ possesses $c$ in the initial state.  The rule (func) states that, if an agent possesses
both a function $f$ and term $t$ of the appropriate types, it may compute the term $ft$.  The rule (message$_1$) states that, after $\alpha$ has sent $t$
to $\beta$, then $\beta$ possesses $t$.  The rule (message$_2$) states that, if $\gamma$ possesses $s$ before $\alpha$ sends a message to $\beta$, then
$\gamma$ still possesses $s$ after the message is sent.

\subsection{Metatheorems}

We can establish the basic properties that our typing system satisfies.

\begin{lemma}$ $
  \begin{enumerate}
    \item \textbf{Weakening}
 Suppose $\tau_1 \vdash \alpha \ni t : A$ and $\tau_1, \tau_2$ is a valid trace.  Then $\tau_1, \tau_2 \vdash \alpha \ni t : A$.
\item
\label{lm:tracevalid}
If $\tau_1, \alpha \stackrel{t : A}{\rightarrow} \beta, \tau_2$ is a valid trace, then $A \in M_{\alpha \beta}$, and $\tau_1 \vdash \alpha \ni t : A$.
\item \textbf{Generation}
\label{lm:gen}
 Suppose $\tau \vdash \beta \ni t : B$.  Then there exist terms $t_1 : A_1$, \ldots, $t_m : A_m$ ($m \geq 0$) and agents $\alpha_1$, \ldots, $\alpha_n$ ($n \geq 1$) such that $t \equiv f t_1 \cdots t_m$, $\beta = \alpha_n$,
 and the following events occur in $\tau$ in order:
 \[ \alpha_1 \ni f : A_1 \rightarrow \cdots \rightarrow A_n \rightarrow B, \; \alpha_1 \stackrel{t : B}{\rightarrow} \alpha_2, \; \cdots, \; \alpha_{n-1} \stackrel{t : B}{\rightarrow} \alpha_n \]
Further, we have $\tau \vdash \alpha_1 \ni t_1 : A_1$, \ldots, $\tau \vdash \alpha_1 \ni t_m : A_m$.
\item
 If $\tau \vdash \beta \ni t : B$, then either $\beta$ can compute $B$, or there is an event $\alpha \stackrel{t : B}{\rightarrow} \beta$ in $\tau$ for some $\alpha$.
\end{enumerate}
\end{lemma}

Intuitively, Generation says that if agent $\beta$ possesses a piece of data of type $B$, then it must have been computed by an agent $\alpha_1$ that can compute terms of type $B$, and then
passed to $\beta$ in a sequence of messages.

The proofs of the first three properties are by straightforward induction on derivations.
Part 4 follows easily from part 3.

\section{The First Algorithm}
\label{section:firstalgo}

In the rest of this paper, we will consider different sets of \emph{constraints} that we may wish to place on our architectures.  In each case, we shall show how, given an architecture $\mathcal{A}$ and a set of constraints
$\mathbb{C}$, we can construct an architecture $\mathcal{B}$, which we call a \emph{safe} architecture, that extends $\mathcal{A}$ and satisfies all the constraints.

For our first algorithm, we consider the following constraints:

\begin{definition}[Constraint]
  \label{df:constraint}
 \begin{enumerate}
  \item \label{df:negcon} A \emph{negative constraint} has the form $\alpha \ni A \Rightarrow B$, where $A$ and $B$ are atomic types.  We read it as: ``If $\alpha$ receives a message of type $A$, then a term of type $B$ must have previously been created.''
  A trace $\tau$ \emph{complies} with this constraint iff, for every $\tau_1 \sqsubseteq \tau$, if $\tau_1 \vdash \alpha \ni t : A$ for some $t$, then $\tau_1 \vdash \beta \ni s : B$ for some $\beta$, $s$.
  \item A \emph{positive constraint} has the form $\mathtt{Pos}(\alpha, A)$, where $A$ is an atomic type.  We read it as: ``It must be possible for $\alpha$ to have a term of type $A$.''  A trace $\tau$ \emph{complies} with this constraint iff $\tau \vdash \alpha \ni t : A$ for some term $t$.
\end{enumerate}
\end{definition}

\paragraph{Note}
To understand part \ref{df:negcon} of this definition, note that, if it is possible to create a term $t : A$ without first creating a term $s : B$,
then there is a trace $\tau$ such that $\tau \vdash \alpha \ni t : A$ for some $\alpha$, and $\tau \nvdash \beta \ni s : B$ for all $\beta$.  Thus, the condition
``For every $\tau_1 \sqsubseteq \tau$, if $\tau_1 \vdash \alpha \ni t : A$ for some $t$, then $\tau_1 \vdash \beta \ni s : B$ for some $\beta$, $s$'' captures
the idea ``If $\alpha$ receives a message of type $A$, then a term of type $B$ must have previously been created.''

\paragraph{Example}
Consider an accountancy firm collecting personal data from the employees of a company in order to prepare a tax report.  The principle of
\emph{data minimization} \cite[Section 25]{eu:gdpr} states that the accountancy firm should collect only the data that is necessary for this purpose.
We can model this as follows: assume there are two types of tax return that can be prepared, $TR_A$ and $TR_B$.  Let $Employee$ initially possess $a : A$ and $b : B$,
where $a$ is required to prepare $TR_A$ and $b$ is required to prepare $TR_B$.  The company can send requests $Q_A$ and $Q_B$ to $Accountancy$, requesting a tax
return of one of the two types.  We could then write constraints $Accountancy \ni A \Rightarrow Q_A$ and $Accountancy \ni B \Rightarrow Q_B$ to express that
the accountancy firm may only possess an employee's personal data if it is necessary for a tax return that it has been requested to prepare.

We now construct the type system that the safe architecture will use:

\begin{definition}[Safe Type System]
 Let $\mathcal{T}$ be a type system and $\Ag$ a set of agents.  Let $\mathbb{C}$ be a finite set of negative constraints over $\mathcal{T}$ and $\Ag$.
 The \emph{safe} type system $\mathcal{T}_{\mathbb{C}}$ is defined as follows.
 \begin{itemize}
  \item The atomic types of $\mathcal{T}_{\mathbb{C}}$ are the atomic types of $\mathcal{T}$ together with, for every agent $\alpha \in \Ag$ and atomic type $A$ in $\mathcal{T}$,
  a type $C_{\alpha}(A)$, the type of \emph{certified} terms of type $A$ that may only be read by $\alpha$.
  \item Every constructor of $\mathcal{T}$ is a constructor of $\mathcal{T}_{\mathbb{C}}$.
  \item For every $\alpha \in \Ag$ and type $A$ of $\mathcal{T}$, let the constraints in $\mathbb{C}$ that begin with `$\alpha \ni A$' be
  \[ \alpha \ni A \Rightarrow B_1, \ldots, \alpha \ni A \Rightarrow B_n \enspace . \]
  Then the following are constructors of $\mathcal{T}_{\mathbb{C}}$:
  \begin{align*}
  & m_{\alpha A}^{\beta_1 \cdots \beta_n} : A \rightarrow C_{\beta_1}(B_1) \rightarrow \cdots \rightarrow C_{\beta_n}(B_n) \rightarrow C_\alpha(A) \text{ for all } \beta_1, \ldots, \beta_n \in \Ag; \\
  & \pi_{\alpha A} : C_\alpha(A) \rightarrow A
  \end{align*}
 \end{itemize}
\end{definition}

The intention is that $m_{\alpha A}^{\beta_1 \cdots \beta_n}$ constructs a term of type $C_\alpha(A)$ out of a term of type $A$ and $n$ other terms which
prove that the preconditions to $\alpha \ni A$ are all satisfied.  The constructor $\pi_{\alpha A}$
then extracts the term of type $A$ again.

Using the type system, we can state a set of conditions that guarantee that an architecture satisfies the negative constraints in $\mathbb{C}$.

\begin{theorem}
\label{theorem:safeone}
 Let $\mathcal{T}$ be a type system, $\Ag$ a set of agents, and $\mathbb{C}$ a set of negative constraints over $\mathcal{T}$ and $\Ag$.
 Let $\mathcal{B}$ be an architecture over $\mathcal{T}_{\mathbb{C}}$ with set of agents $\Ag'$, where $\Ag \subseteq \Ag'$.  Suppose
there is a partition $\{ \mathcal{P}_\alpha \subseteq \Ag' \}_{\alpha \in \Ag}$ of $\Ag'$ indexed by $\Ag$ such that:
\begin{enumerate}
\item
$\alpha \in \mathcal{P}_\alpha$ for all $\alpha \in \Ag$;
\item
If $\beta \stackrel{A}{\rightarrow} \beta'$ and $\beta$, $\beta'$ are in different sets of the partition, then $A$ has the form $C_\gamma(B)$ for some $\gamma$, $B$;
\item
If $\beta$ initially possesses $\pi_{\alpha A}$ then $\beta \in \mathcal{P}_\alpha$;
\item
For every constraint $\alpha \ni A \Rightarrow B$ in $\mathbb{C}$, if an agent $\beta \in \mathcal{P}_\alpha$ possesses a constructor with target $A$, then this constructor
is $\pi_{\alpha A}$.
 \end{enumerate}
Then every trace through $\mathcal{B}$ satisfies every negative constraint in $\mathbb{C}$.
\end{theorem}

The intuition behind the premises is this: the partition divides the system into parts.  The part $\mathcal{P}_\alpha$ is the only part of the system that is allowed to
look inside a term of type $C_\alpha(A)$ and extract the underlying term of type $A$.  Only certified terms may be passed between the parts.  Thus, the only way
for an agent in $\mathcal{P}_\alpha$ to possess a term of type $A$ is either for it to be computed within $\mathcal{P}_\alpha$, or for a term of type $C_\alpha(A)$
to be passed in from another part of the system.

\begin{proof}
 Let $\tau$ be any trace through $\mathcal{B}$ and let $\alpha \ni A \Rightarrow B$ be one of the constraints in $\mathbb{C}$.  We must show that, if $\tau \vdash \alpha \ni t : A$,
 then $\tau \vdash B$.  We shall prove the more general result:
 \begin{quote}
  If $\tau \vdash \beta \ni t : A$ for some $\beta \in \mathcal{P}_\alpha$, then $\tau \vdash B$.
 \end{quote}
So suppose $\tau \vdash \beta \ni t : A$ for some $\beta \in \mathcal{P}_\alpha$.  We may also assume without loss of generality that $\tau$ is the shortest trace for which
this is true.  By Generation and the minimality of $\tau$, $\beta$ possesses a constructor with target $A$.  By our hypotheses, this is $\pi_{\alpha A}$, and $t = \pi_{\alpha A}(t')$
for some $t'$.  Hence $\tau \vdash \beta \ni t' : C_{\alpha}(A)$
for some $t'$.

Now, looking at the construction of $\mathcal{T}_{\mathbb{C}}$, the only constructor with target $C_{\alpha}(A)$ is
\[ m_{\alpha A}^{\beta_1 \cdots \beta_n} : A \rightarrow C_{\beta_1}(B_1) \rightarrow \cdots \rightarrow C_{\beta_n}(B_n) \rightarrow C_{\alpha}(A) \enspace . \]
So applying Generation again, we must have $t \equiv m_{\alpha A}^{\beta_1 \cdots \beta_n} s t_1 \cdots t_n$ and there must be an agent $\gamma$
which possesses $m_{\alpha A}^{\beta_1 \cdots \beta_n}$ with
\[ \tau \vdash \gamma \ni s : A, \quad \tau \vdash \gamma \ni t_1 : B_1, \ldots, \tau \vdash \gamma \ni t_n : B_n \enspace . \]
Now, $B$ is one of the types $B_1, \ldots, B_n$; let it be $B_i$.  Then $\tau \vdash \gamma \ni C_{\beta_i}(B)$.
By similar reasoning, there must be an agent $\delta$ that possesses one of the constructors $m_{\beta_i B}$, and
$\tau \vdash \delta \ni B$. $\qed$
\end{proof}

We are now ready to construct the safe architecture.

\begin{algorithm}
 Given an architecture $\mathcal{A}$ and a finite set of constraints $\mathbb{C}$, construct the architecture $\mathtt{Safe}(\mathcal{A}, \mathbb{C})$ as follows:
 \begin{enumerate}
  \item The agents of $\mathtt{Safe}(\mathcal{A}, \mathbb{C})$ are the agents of $\mathcal{A}$ together with, for every agent $\alpha$ of $\mathcal{A}$, an agent $I_\alpha$, which we call the \emph{interface} to $\alpha$.
  \item The type system of $\mathtt{Safe}(\mathcal{A}, \mathbb{C})$ is $\mathcal{T}_\mathbb{C}$.
  \item If an agent $\alpha$ possesses a constructor $c$ in $\mathcal{A}$, then $\alpha$ possesses $c$ in $\mathtt{Safe}(\mathcal{A}, \mathbb{C})$.
   \item For every type $A$ of $\mathcal{A}$, let the negative constraints that begin with $\alpha \ni A$ be
   \[ \alpha \ni A \Rightarrow B_1, \ldots, \alpha \ni A \Rightarrow B_n \enspace . \]
   \begin{itemize}
   \item
   \emph{Every} interface $I_\gamma$ possesses $m_{\alpha A}^{\beta_1 \cdots \beta_n}$ for all $\beta_1$, \ldots, $\beta_n$.
   \item
   $I_\alpha$ posseses $\pi_{\alpha A}$
   \end{itemize}
  \item For every atomic type $A$, the agents $\alpha$ and $I_\alpha$ may send $A$ to each other.
  \item Any two interfaces may send messages of type $C_\alpha(A)$ to each other for any $\alpha$, $A$.
 \end{enumerate}
\end{algorithm}

Thus, in order to construct a certified term of type $A$ readable by $\alpha$, an interface must first obtain certified terms of all the types which the constraints require.  The only way $\alpha$ can
receive a term of type $A$ is through its interface obtaining a term of type $C_\alpha(A)$.  Interfaces may pass certified terms between each other at will.  An agent and its interface may exchange uncertified terms at will.

\begin{theorem}
  \label{theorem:safetwo}
 Let $\mathcal{A}$ be an architecture and $\mathbb{C}$ a set of constraints.  Suppose that:
 \begin{enumerate}
  \item For every negative constraint $\alpha \ni A \Rightarrow \beta \ni B$ in $\mathbb{C}$, we have that $\alpha$ cannot compute terms of type $A$.
  \item For every positive constraint $\mathtt{Pos}(\alpha, A) \in \mathbb{C}$, there exists a trace through $\mathcal{A}$ that satisfies $\mathtt{Pos}(\alpha, A)$ and all the negative constraints in $\mathbb{C}$.
  \end{enumerate}
 Then the architecture $\mathtt{Safe}(\mathcal{A}, \mathbb{C})$ has the following properties:
 \begin{enumerate}
  \item Every trace through $\mathtt{Safe}(\mathcal{A}, \mathbb{C})$ satisfies every negative constraint in $\mathbb{C}$.
  \item For every positive constraint $\mathtt{Pos}(\alpha, A) \in \mathbb{C}$, there exists a trace through $\mathtt{Safe}(\mathcal{A}, \mathbb{C})$ that satisfies $\mathtt{Pos}(\alpha, A)$.
 \end{enumerate}
\end{theorem}

\begin{proof}
Part 1 follows from the previous theorem, taking $\mathcal{P}_\alpha = \{ \alpha, I_\alpha \}$.

We now show that $\mathtt{Safe}(\mathcal{A}, \mathbb{C})$ has the following property.  Part 2 of the theorem follows immediately.
\begin{quote}
 If $\tau \vdash \alpha \ni t : A$ in $\mathcal{A}$, $A$ is an atomic type, and $\tau$ satisfies every negative constraint in $\mathbb{C}$, then there exists a valid trace $\tau'$ through $\mathtt{Safe}(\mathcal{A}, \mathbb{C})$ such
 that $\tau' \vdash \alpha \ni t : A$ and $\tau' \vdash I_\alpha \ni t' : C_\alpha(A)$ for some $t'$.
\end{quote}

The proof is by induction on $\tau$, then on the derivation of $\tau \vdash \alpha \ni t : A$.  We deal here with the case where the last rule in the derivation was $(message_1)$:
\[ \infer{\tau, \beta \stackrel{t : A}{\rightarrow} \alpha \vdash \alpha \ni t : A}{\tau \vdash \beta \ni t : A} \enspace . \]
By the induction hypothesis, there exists $\tau'$ such that
$\tau' \vdash_{\mathtt{Safe}(\mathcal{A}, \mathbb{C})} \beta \ni t : A$.
By the construction of $\mathtt{Safe}(\mathcal{A}, \mathbb{C})$, we have $A \in C_{\beta I_\beta}$ and $A \in C_{I_\alpha \alpha}$.  Hence
$\tau, \beta \stackrel{t : A}{\rightarrow} I_\beta \vdash_{\mathtt{Safe}(\mathcal{A}, \mathbb{C})} I_\beta \ni t : A$.

Now, let the negative constraints in $\mathbb{C}$ that begin with $\alpha \ni A$ be
$\alpha \ni A \Rightarrow B_1, \ldots, \alpha \ni A \Rightarrow B_n$.
By hypothesis, $\tau, \beta \stackrel{t : A}{\rightarrow} \alpha$ satisfies all these constraints.  Therefore, $\tau \vdash_\mathcal{A} B_1, \ldots, \tau \vdash_\mathcal{A} B_n$.

Hence, by the induction hypothesis, there exists $\tau''$ such that
$\tau'' \vdash_{\mathtt{Safe}(\mathcal{A}, \mathbb{C})} B_1, \ldots, \tau'' \vdash_{\mathtt{Safe}(\mathcal{A}, \mathbb{C})} B_n$.
Therefore,
\[ \tau'' \vdash_{\mathtt{Safe}(\mathcal{A}, \mathbb{C})} I_{\beta_1} \ni t_1 : C_{\beta_1}(B_1), \ldots, \tau'' \vdash_{\mathtt{Safe}(\mathcal{A}, \mathbb{C})} I_{\beta_n} \ni t_n : C_{\beta_n}(B_n) \enspace , \]
for some $t_1$, \ldots, $t_n$.
By Weakening, we may assume $\tau' \sqsubseteq \tau''$.

After extending $\tau''$ by passing $t_1$, \ldots, $t_n$ as messages to $I_\beta$, we have that $I_\beta$ can construct a term of type $C_\alpha(A)$.  After passing this term to $I_\alpha$, we have that $I_\alpha(A)$ possesses a term of
type $C_\alpha(A)$.  From this, it can construct a term of type $A$ which it may then pass to $\alpha$, completing the required trace. $\qed$
\end{proof}

\section{The Second Algorithm}
\label{section:secondalgo}

Supposing it is important to us, not merely that a piece of data has been created, but that a particular agent has seen it.  We can extend our system to handle this
type of constraint as follows.

\begin{definition}
  \label{df:constraint2}
In this section of the paper:
\begin{itemize}
 \item a \emph{negative constraint} is an expression of the form $\alpha \ni A \Rightarrow \beta \ni B$.  A trace $\tau$ satifies this constraint iff, for every $\tau' \sqsubseteq \tau$, if $\tau' \vdash \alpha \ni A$ then
 $\tau_1 \vdash \beta \ni B$.
 \item Positive constraints are as in Section \ref{section:firstalgo}.
\end{itemize}
\end{definition}

\paragraph{Note}
If $(\alpha, A) \neq (\beta, B)$, then the constraint $\alpha \ni A \Rightarrow \beta \ni B$
is to be read as ``if $\alpha$ possesses a term of type $A$, then $\beta$ must previously have possessed a term of type $B$''.  (The condition $\alpha \ni A \Rightarrow \alpha \ni A$
is trivial.)
 
We show how to extend a given architecture $\mathcal{A}$ to
an architecture that uses the new privacy-safe type system.  Unfortunately, we have not found a way to do this that requires no modifications to the agents in $\mathcal{A}$.
We present below (Algorithm 2) an algorithm that requires modifications which we expect would be minor in practice, and discuss
in Section \ref{section:note} ways in which this situation could be improved in future work.

\begin{definition}
 Given a type system $\mathcal{T}$, a set of agents $\Ag$, and a set of negative constraints $\mathbb{C}$ over $\mathcal{T}$ and $\Ag$,
 define the type system $\mathcal{T}_\mathbb{C}$ as follows.
 \begin{itemize}
  \item The types of $\mathcal{T}_\mathbb{C}$ are the types of $\mathcal{T}$ together with, for every agent $\alpha$ and atomic type $A$ of $\mathcal{T}$, a type $C_\alpha(A)$
  and a type $P_\alpha(A)$.  (Intuition: a term $C_\alpha(A)$ is a \emph{certified} term of type $A$ that $\alpha$ is permitted to read.  A term $P_\alpha(A)$ is
  \emph{proof} that $\alpha$ has held a term of type $A$.)
  \item Every constructor of $\mathcal{T}$ is a constructor of $\mathcal{T}_\mathbb{C}$.
  \item For every agent $\alpha$ and type $A$, let the negative contraints in $\mathbb{C}$ that begin with $\alpha \ni A$ be
  \[ \alpha \ni A \Rightarrow \beta_1 \ni B_1, \ldots, \alpha \ni A \Rightarrow \beta_n \ni B_n \enspace . \]
  Then the following are constructors of $\mathcal{T}_\mathbb{C}$:
  \begin{align*}
 m_{\alpha A} & : A \rightarrow P_{\beta_1}(B_1) \rightarrow \cdots \rightarrow P_{\beta_n}(B) \rightarrow C_\alpha(A) \\
 \pi_{\alpha A} & : C_\alpha(A) \rightarrow A \\
 p_{\alpha A} & : A \rightarrow P_\alpha(A)
 \end{align*}
 \end{itemize}
\end{definition}

\begin{theorem}
\label{theorem:two}
 Let $\mathcal{T}$ be a type system, $\Ag$ a set of agents, and $\mathbb{C}$ a set of negative constraints over $\mathcal{T}$ and $\Ag$.
 Let $\mathcal{B}$ be an architecture over $\mathcal{T}_\mathbb{C}$ with set of agents $\Ag'$, where $\Ag \subseteq \Ag'$.  Suppose that there is a partition $\{ \mathcal{P}_\alpha \}_{\alpha \in \Ag}$
 of the agents of $\mathcal{B}$ such that:
 \begin{itemize}
  \item $\alpha \in \mathcal{P}_\alpha$;
  \item If $\beta \stackrel{A}{\rightarrow} \beta'$ and $\beta$ and $\beta'$ are in different sets in the partition, then $A$ has either the form $C_\gamma(T)$ or $P_\gamma(T)$;
  \item If $\beta$ initially possesses $\pi_{\alpha A}$ then $\beta \in \mathcal{P}_\alpha$;
  \item If $\beta$ initially possesses $p_{\alpha A}$ then $\beta$ cannot compute $A$.
  \item If $\beta$ initially possesses $p_{\alpha A}$ and $\gamma \stackrel{A}{\rightarrow} \beta$ then $\gamma = \alpha$.
 \end{itemize}
 Then every trace through $\mathcal{B}$ satisfies every constraint in $\mathbb{C}$.
\end{theorem}

\begin{proof}
Let $\tau$ be a trace through $\mathcal{B}$ and $\alpha \ni A \Rightarrow \beta \ni B$ be a constraint in $\mathbb{C}$.  We must show that, if $\tau \vdash \alpha \ni A$, then
$\tau \vdash \beta \ni B$.  We shall prove the more general result:
\begin{quote}
 If $\tau \vdash \gamma \ni A$ for any $\gamma \in \mathcal{P}_\alpha$, then $\tau \vdash \beta \ni B$.
\end{quote}
So suppose $\tau \vdash \gamma \ni A$ for some $\gamma \in \mathcal{P}_\alpha$.  We may assume without loss of generality that $\tau$ is the shortest such trace.  By Generation and
the minimality of $\tau$, $\gamma$ must possess a constructor with target $A$.  By our hypotheses, this is $\pi_{\alpha A}$.  Hence $\tau \vdash \gamma \ni t : C_{\alpha}(A)$
for some $t$.  Now, let the constraints in $\mathbb{C}$ that begin with $\alpha \ni A$ be
\[ \alpha \ni A \Rightarrow \beta_1 \ni B_1, \quad \cdots, \quad \alpha \ni A \Rightarrow \beta_n \ni B_n \enspace . \]
Applying Generation, we must have $t \equiv m_{\alpha A} s t_1 \cdots t_n$, and there must be an agent $\gamma'$ that possesses $m_{\alpha A}$ such that
\[ \tau \vdash \gamma' \ni s : A, \quad \tau \vdash \gamma' \ni t_1 : P_{\beta_1}(B_1), \ldots, \tau \vdash \gamma' \ni t_n : P_{\beta_n}(B_n) \enspace . \]
Now, there is some $i$ such that $\beta_i = \beta$ and $B_i = B$.  We have $\tau \vdash \gamma' \ni t_i : P_{\beta}(B)$.
Since a term of type $P_{\beta}(B)$ has been constructed, it must be that $\tau \vdash \beta \ni B$, as required.
\end{proof}

We now show again how, given an architecture $\mathcal{A}$, we can construct an architecture that is privacy-safe.

\begin{algorithm}
 Given an architecture $\mathcal{A}$ and a finite set of constraints $\mathbb{C}$, construct the architecture $\mathtt{Safe}(\mathcal{A}, \mathbb{C})$ as follows:
 \begin{enumerate}
  \item The agents of $\mathtt{Safe}(\mathcal{A}, \mathbb{C})$ are the agents of $\mathcal{A}$ together with, for every agent $\alpha$ of $\mathcal{A}$:
  \begin{itemize}
   \item an agent $I_\alpha$, which we call the \emph{input interface} to $\alpha$;
   \item an agent $O_\alpha$, which we call the \emph{output interface} to $\alpha$
  \end{itemize}
  \item The type system of $\mathtt{Safe}(\mathcal{A}, \mathbb{C})$ is $\mathcal{T}_\mathbb{C}$.
  \item If an agent $\alpha$ has a constructor $c$ in $\mathcal{T}$, then it has the constructor $c$ in $\mathcal{T}_\mathbb{C}$.
  \item For any agent $\alpha$ and type $A$:
  \begin{itemize}
   \item Every output interface $O_\gamma$ possesses $m_{\alpha A}$
   \item $I_\alpha$ possesses $\pi_{\alpha A} : C_{\alpha A} \rightarrow A$
   \item $O_\alpha$ possesses $p_{\alpha A} : A \rightarrow P_{\alpha A}$
  \end{itemize}
  \item
  For any atomic type $A$ of $\mathcal{T}$, $I_\alpha$ may send $A$ to $\alpha$, and $\alpha$ may send $A$ to $O_\alpha$.
  \item Any two interfaces may send messages of type $C_\alpha(A)$ or $P_\alpha(A)$ to each other for any $\alpha$, $A$.
  \end{enumerate}
\end{algorithm}

\begin{theorem}
 Let $\mathcal{A}$ be an architecture and $\mathbb{C}$ a set of constraints.  Suppose that:
 \begin{enumerate}
  \item For every negative constraint $\alpha \ni A \Rightarrow B$ in $\mathbb{C}$, we have that $\alpha$ cannot compute terms of type $A$.
  \item For every positive constraint $\mathtt{Pos}(\alpha, A) \in \mathbb{C}$, there exists a trace through $\mathcal{A}$ that satisfies $\mathtt{Pos}(\alpha, A)$ and all the negative constraints in $\mathbb{C}$.
  \end{enumerate}
 Then the architecture $\mathtt{Safe}(\mathcal{A}, \mathbb{C})$ has the following properties:
 \begin{enumerate}
  \item Every trace through $\mathtt{Safe}(\mathcal{A}, \mathbb{C})$ satisfies every negative constraint in $\mathbb{C}$.
  \item For every positive constraint $\mathtt{Pos}(\alpha, A) \in \mathbb{C}$, there exists a trace through $\mathtt{Safe}(\mathcal{A}, \mathbb{C})$ that satisfies $\mathtt{Pos}(\alpha, A)$.
 \end{enumerate}
\end{theorem}

\begin{proof}
Part 1 follows from Theorem \ref{theorem:two}, taking $P_\alpha = \{ \alpha, I_\alpha, O_\alpha \}$.

We shall now prove the following property, from which part 2 of the theorem follows.
\begin{quote}
 If $\tau \vdash \alpha \ni t : A$ in $\mathcal{A}$ and $\tau$ satisfies every negative constraint in $\mathbb{C}$, then there exists
 a trace $\tau'$ through $\mathtt{Safe}(\mathcal{A}, \mathbb{C})$ such that $\tau' \vdash \alpha \ni t : A$.
\end{quote}

The proof is by induction on $\tau$, then on the derivation of $\tau \vdash \alpha \ni t : A$.  We deal here with the case where the final step in the derivation is an instance of $(message_1)$:
\[ \infer{\tau, \beta \stackrel{t : A}{\rightarrow} \alpha \vdash \alpha \ni t : A}{\tau \vdash \beta \ni t : A} \]
By the induction hypothesis, there is a trace $\tau'$ such that $\tau' \vdash_{\Safe{\mathcal{A}}{\mathbb{C}}} \alpha \ni A$.
Let the negative constraints beginning with $\alpha \ni A$ be
\[ \alpha \ni A \Rightarrow \beta_1 \ni B_1, \quad \ldots, \quad \alpha \ni A \Rightarrow \beta_n \ni B_n \enspace . \]
Then, by hypothesis,
\[ \tau, \beta \stackrel{t : A}{\rightarrow} \alpha \vdash_{\mathcal{A}} \beta_1 \ni B_1, \quad \cdots, \quad \tau, \beta \stackrel{t : A}{\rightarrow} \alpha \vdash_{\mathcal{A}} \beta_n \ni B_n \enspace . \]
Using the fact that $(\alpha, A) \neq (\beta_i, B_i)$ for all $i$, the last step in each of these derivations must have been $(message_2)$.  Therefore,
\[ \tau \vdash_{\mathcal{A}} \beta_1 \ni B_1, \quad \cdots, \quad \tau \vdash_{\mathcal{A}} \beta_n \ni B_n \enspace . \]
We may therefore apply the induction hypothesis to obtain traces $\tau_1$, \ldots, $\tau_n$ such that
\[ \tau_1 \vdash_{\Safe{\mathcal{A}}{\mathbb{C}}} \beta_1 \ni t_1 : B_1, \quad \cdots, \quad \tau_n \vdash_{\Safe{\mathcal{A}}{\mathbb{C}}} \beta_n \ni t_n : B_n \enspace . \]
Now, let $\tau''$ be the trace $\tau', \tau_1, \ldots, \tau_n$ followed by these events:
\begin{align*}
 & \beta \stackrel{t : A}{\longrightarrow} O_\beta, \beta_1 \stackrel{t_1 : B_1}{\longrightarrow} O_{\beta_1}, \cdots, \beta_n \stackrel{t_n : B_n}{\longrightarrow} O_{\beta_n}, \\
 & O_{\beta_1} \stackrel{p_{\beta_1 B_1} t_1}{\longrightarrow} O_\beta, \cdots, O_{\beta_n} \stackrel{p_{\beta_n, B_n} t_n}{\longrightarrow} O_\beta, \\
 & O_\beta \stackrel{c_{\alpha A} t (p_{\beta_1 B_1} t_1) \cdots (p_{\beta_n B_n} t_n)}{\longrightarrow} I_\alpha, \\
 & I_\alpha \stackrel{\pi_{\beta A}(c_{\alpha A} t (p_{\beta_1 B_1} t_1) \cdots (p_{\beta_n B_n} t_n))}{\longrightarrow} \alpha
\end{align*}
(Informally: the agent $O_\beta$ collects the term of type $A$ from $\beta$ and all the necessary proofs, assembles the term of type $C_{\alpha A}$, and passes
it to $I_\alpha$, who decodes it with $\pi_{\alpha A}$ and passes the value of $A$ to $\alpha$.)

We thus have $\tau'' \vdash \alpha \ni A$ in $\Safe{\mathcal{A}}{\mathbb{C}}$, as required. \qed
\end{proof}

\subsection{Example Revisited}
\label{section:revisit}

We return to the example we presented in Section \ref{section:example}.  We are now ready to formulate our third, positive constraint.  We want to ensure it is possible for
the website to receive the child's information once all legal requirements have been met.  So the privacy constraints that we require for this architecture are:
\begin{description}
 \item[Negative Constraint] $\mathit{Website} \ni \mathit{INFO} \Rightarrow \mathit{Website} \ni \mathit{CONSENT}$
 \item[Negative Constraint] $\mathit{Website} \ni \mathit{CONSENT} \Rightarrow \mathit{Parent} \ni \mathit{POLICY}$
 \item[Positive Constraint] $\Pos{\mathit{Website}}{\mathit{INFO}}$
\end{description}

We can verify that the first constraint holds.  The child can send the protected info to the interface $OChild$, but it cannot then be sent to another agent unless
$OChild$ receives a term of type $P(CONSENT)$.  And for a term of type $P(CONSENT)$ to be constructed, the parent must have sent consent to the website (via $OParent$ and
$IWebsite$).

We can also verify that, in the architecture in Figure \ref{fig:coppa2}, it is possible for the website to send the privacy policy to the parent, the parent to send consent to
the website, and the child to send the protected info to the website.  Formally, we describe a valid trace $\tau$ through the architecture that represents this sequence
of events.  The trace $\tau$ begins
\[ \begin{array}{rcl}
Website & \stackrel{policy : POLICY}{\longrightarrow} & OWebsite, \\
\\
 OWebsite & \stackrel{mPOLICY(policy): C(POLICY)}{\longrightarrow} & IParent, \\
 \\
 IParent & \stackrel{\pi POLICY(mPOLICY(policy)) : POLICY}{\longrightarrow} & Parent, \\
 \\
 Parent & \stackrel{consent : CONSENT}{\longrightarrow} & OParent, \\
 \\
 Parent & \stackrel{\pi POLICY(mPOLICY(policy)) : POLICY}{\longrightarrow} & OParent
\end{array} \]
Let $p = \pi POLICY(mPOLICY(policy)$.  The trace $\tau$ continues:
\[ \begin{array}{rcl}
 OParent & \stackrel{mCONSENT(consent, p) : C(CONSENT)}{\longrightarrow} & IWebsite, \\
 \\
 IWebsite & \stackrel{\pi CONSENT(mCONSENT(consent, p)) : CONSENT}{\longrightarrow} & Website, \\
 \\
 Website & \stackrel{\pi CONSENT(mCONSENT(consent, p)) : CONSENT}{\longrightarrow} & OWebsite, \\
\end{array} \]
Let $c = \pi CONSENT(mCONSENT(consent, p))$.  The trace $\tau$ continues:
\[ \begin{array}{rcl}
 OWebsite & \stackrel{pCONSENT(c) : P(CONSENT)}{\longrightarrow} & OChild, \\
 \\
 Child & \stackrel{info : INFO}{\longrightarrow} & OChild, \\
 \\
 OChild & \stackrel{mINFO(info, pCONSENT(c)) : C(INFO)}{\longrightarrow} & IWebsite, \\
 \\
 IWebsite & \stackrel{\pi INFO(mINFO(info, pCONSENT(c))): INFO}{\longrightarrow} & Website
\end{array} \]
This ends the trace $\tau$ which verifies that it is possible for $Website$ to receive a term of type $INFO$.

\subsection{Note}
\label{section:note}

In Fig. \ref{fig:coppa2}, we have had to modify the agents from Fig. \ref{fig:coppa1}.  The agent $\mathit{Parent}$ needs to be able to output
messages of type $\mathit{POLICY}$, and $\mathit{Website}$ needs to be able to output messages of type $\mathit{CONSENT}$.
We believe these would be minor changes in practice.  However, this is still unfortunate, because as discussed in the Introduction, we want our
algorithms to apply in cases in which we are unable to change the source code of the agents in $\mathcal{A}$.

In practice, we could implement this by allowing $\mathit{IParent}$ to send $\mathit{POLICY}$ to $\mathit{OParent}$, and $\mathit{IWebsite}$ to send $\mathit{POLICY}$ to
$\mathit{OWebsite}$, and adding the following local constraints to their behaviour:

\begin{itemize}
 \item If $\mathit{IParent}$ sends $t : \mathit{POLICY}$ to $\mathit{OParent}$, then $\mathit{IParent}$ must previously have sent $t$ to $\mathit{Parent}$.
 \item If $\mathit{IWebsite}$ sends $t : \mathit{POLICY}$ to $\mathit{OWebsite}$, then $\mathit{IWebsite}$ must previously have sent $t$ to $\mathit{Website}$.
\end{itemize}

Obtaining a formal proof of correctness for this construction requires an architecture language in which this sort of local constraint can be expressed, and we leave this for future work.

\section{Related Work}
\label{section:related}

Le M\'{e}tayer et al. \cite{pbd,privacyarchitectures,rdp,lda} have described several languages for describing architectures and deciding privacy-related properties
over them.  Barth et al. \cite{contextualintegrity} also give a formal definition of architectures, and show how to decide properties defined in temporal logic.  Our work
was heavily inspired by these systems; however, our aim was to give a method to design an architecture starting from a set of privacy properties, and not to decide whether a property
holds of a given architecture.

Basin et al. \cite{basinetal} show how to describe privacy policies in metric first-order temporal logic (MFOTL), and how to build a component that monitors in real-time whether 
these policies are being violated.  Nissenbaum et al \cite{contextualintegrity} also describe privacy policies using linear temporal logic (LTL), and this has inspired a lot
of research into systems such as P-RBAC, which enforces low-level privacy-related conditions at run-time \cite{ni2010privacy}.  Most of this research has concentrated on verifying
at run-time whether or not a given action is permitted by a given set of privacy policies.  The work presented here concentrates instead on design-time, and ensures that a high-level privacy
policy is followed, no matter what actions each individual component performs with the data it receives, as long as all messages follow the given type system.

Jeeves \cite{YangYS12} is a constraint functional language motivated by
separating  business logic and confidentially concerns. We could
implement our (architectural) constraints in Jeeves, but would no longer
have static guarantees. Other work in formal methods for privacy
includes static analysis of source code
\cite{cortesi2015datacentric,ferrara2015morphdroid}
and refinement techniques for deriving low-level designs from high-level
designs in a way that preserves privacy properties
\cite{alur2006preserving,clarkson2010hyperproperties}.
These approaches complement ours well, addressing properties for
individual components that cannot be expressed in our constraint
language, while our
algorithms provide formal guarantees of global properties of the system
as a whole.

Other work in formal methods for privacy has tended to concentrate either on static analysis of source code \cite{ferrara2015morphdroid,cortesi2015datacentric}
or on refinement techniques for deriving low-level designs from high-level designs in a way that preserves privacy properties \cite{alur2006preserving,clarkson2010hyperproperties}.
These approaches should complement ours well, providing formal guarantees for individual components of properties that cannot be expressed in our constraint language, while our
algorithms provide formal guarantees of global properties of the system as a whole.

\section{Conclusion}
\label{section:conclusion}

We have given two algorithms which take an architecture, and a set of constraints on that architecture, and show how the architecture may be extended in such a way that we can produce a formal proof that the negative constraints hold on every trace through the architecture, and the positive constraints are satisfiable.
Moreover, we do not need to read or modify the source code of the components from the original architecture in order to do this.
We believe this is a promising approach to designing large, complex systems, with many different parts designed and maintained by different people,
such that we can provide a formal proof of privacy-relevant properties.

For the future, we wish to expand the language that may be used for our constraints, for example by allowing the designer to express constraints using propositional, predicate
or temporal logic.  We hope then to express other properties that are desirable for privacy, such as the obligation to delete data.  This will require in turn expanding our type systems $\mathcal{T}_{\mathbb{C}}$.  We also plan to construct a prototype implementation of the interfaces
described in this paper.

\bibliography{privacy_by_design}

\end{document}